\documentclass[a4paper]{amsart}
\usepackage{bbm}

\newtheorem{theorem}{Theorem}

\def\cz{\mathbbm{C}}
\def\rz{\mathbbm{R}}

\def\cE{\mathcal{E}}
\def\cI{\mathcal{I}}

\def\dx{\mathrm{d}x}
\def\dy{\mathrm{d}y}

\def\tr{\mathrm{tr}}
\def\const{\mathrm{const}\;}
\def\etf{E^{\mathrm{TF}}(Z)}

\begin{document}

\title[Atomic Density of the Chandrasekhar Hamiltonian]{The Atomic Density on
  the Thomas--Fermi Length Scale for the Chandrasekhar Hamiltonian}

\author[K. Merz]{Konstantin Merz}
\email{merz@math.lmu.de}

\author[H. Siedentop]{Heinz Siedentop} \email{h.s@lmu.de}
\address{Mathematisches Institut\\ Ludwig-Maximilians-Universit\"at
  M\"unchen\\ Theresienstra\ss e 39\\ 80333 M\"unchen\\ Germany}

\begin{abstract}
  We consider a large neutral atom of atomic number $Z$, modeled by a
  pseudo-relativistic Hamiltonian of Chandrasekhar.  We study its
  suitably rescaled one-particle ground state density on the
  Thomas--Fermi length scale $Z^{-1/3}$. Using an observation by
  Fefferman and Seco \cite{FeffermanSeco1989}, we find that the
  density on this scale converges to the minimizer of the Thomas--Fermi
  functional of hydrogen as $Z\to\infty$ when $Z/c$ is fixed
  to a value not exceeding $2/\pi$. This shows that the electron
  density on the Thomas--Fermi length scale does not exhibit any
  relativistic effects.
\end{abstract}

\maketitle
\section{Introduction}
The energy of heavy atoms as well as the distribution of its electrons
are of fundamental interest both in physics and in quantum
chemistry. However, as in the classical Kepler problem, one cannot
hope for an exact solution of the Schr\"odinger equation involving
more than two particles. For this reason, one needs to devise models
for many-body quantum systems which are easier to solve but still
describe the system accurately.

Lieb and Simon \cite{LiebSimon1977} showed that the atomic ground
state density converges on the length scale $Z^{-1/3}$ to the
minimizer of the Thomas--Fermi functional of hydrogen. This results is
derived by controlling the atomic energy to leading order in $Z$ and
its derivative with respect to small perturbations.

However, it is questionable to describe large $Z$ atoms
non-relativistically, since the large nuclear charge forces the bulk
of the electrons on orbits on the length scale $Z^{-1/3}$ from the
nucleus. Thus, electrons close to the nucleus are moving faster than a
substantial fraction of the velocity of light $c$. This suggests that a
relativistic description is necessary.

On the other hand, S\o rensen \cite{Sorensen2005} showed in the context
of the simplest relativistic model, namely the Chandrasekhar operator,
that energetically this worry is not justified, at least not to
leading order in the energy: the atomic ground state energy of the
Chandrasekhar operator is still described by the Thomas--Fermi energy
for large $Z$ and $\gamma:=Z/c$ fixed to a value not exceeding the
critical coupling constant $\gamma_c:=2/\pi$. A similar result for the
Brown--Ravenhall operator was proven by Cassanas et al
\cite{CassanasSiedentop2006}.

Schwinger \cite{Schwinger1981} predicted that relativistic effects
occur only in sub-leading order. Frank et al \cite{Franketal2008} and
Solovej et al \cite{Solovejetal2008} showed, using completely
different approaches, that this is indeed the case.  In particular,
the authors showed that the coefficient of this order is less than the
non-relativistic one which reflects the fact that the relativistic
kinetic energy is lower than the non-relativistic one, especially for
high momenta.

The question arises whether the density on the Thomas--Fermi length
scale $Z^{-1/3}$ is also unchanged by relativistic effects. This might
be conjectured, since the leading energy correction is generated by
the fast electrons close to the nucleus. Our main result is a positive
answer to this question: we show that the suitably rescaled density of
the atomic Chandrasekhar operator converges for large $Z$ and $\gamma$
fixed to a value not exceeding $\gamma_c$ to the minimizer of the
Thomas--Fermi functional of hydrogen.

\section{Definition and main result}
Our system consists of a neutral atom, i.e., a nucleus of charge $Z$
located at the origin with $N=Z$ electrons with $q$ spin states whose
motion is described by the Chandrasekhar operator. It is given by the
Friedrichs extension of the quadratic form associated to
\begin{equation}\label{C}
C_{c,Z}:=\sum_{\nu=1}^N\left(\sqrt{-c^2\Delta_\nu+c^4}-c^2-\frac{Z}{|x_\nu|}\right)
+\sum_{1\leq \nu<\mu\leq N}\frac{1}{|x_\nu-x_\mu|}
\end{equation}
in the Fermionic Hilbert space $\bigwedge_{\nu=1}^N(L^2(\rz^3):\cz^q).$
(Throughout we use atomic units, i.e., $\hbar=e=m=1$.) The constant $c$
denotes the velocity of light which in these units is the inverse of
Sommerfeld's fine-structure constant $\alpha$. Here we focus on
$N=Z$.  The form is bounded from below, if and only if
$\gamma\leq\gamma_c$ (Kato \cite[Chapter Five, Equation
(5.33)]{Kato1966}, Herbst \cite[Theorem 2.5]{Herbst1977}, Weder
\cite{Weder1974}). For $\gamma<\gamma_c$, its form domain
is $H^{1/2}(\rz^{3N}:\cz^{q^N})\cap\bigwedge_{\nu=1}^N(L^2(\rz^3:\cz^q))$ by the KLMN
theorem. In fact, Hardekopf and Sucher \cite{HardekopfSucher1985}
indicated numerically and gave arguments and Raynal et al
\cite{Raynaletal1994} showed that the one-particle operator is
strictly bigger than $-1$, even for $\gamma=\gamma_c$.

A general fermionic ground state density matrix can be written as
$$\sum_{\mu=1}^M w_\mu\lvert\psi_\mu\rangle\langle\psi_\mu\rvert$$
where the $\psi_\mu$ constitute an orthonormal basis of the ground
state eigenspace and the $w_\mu$ are non-negative weights such that
$\sum_{\mu=1}^M w_\mu=1$.  The corresponding one-particle density
$\rho$ is given by
$$
\rho_Z(x)
:=N\sum_{\mu=1}^Mw_\mu\sum_{\sigma_1,...,\sigma_N=1}^q\int_{\rz^{3(N-1)}}|\psi_\mu(x,\sigma_1;x_2,\sigma_2;...;x_N,\sigma_N)|^2\dx_2...\dx_N.
$$
The ground state energy of this system for fixed $\gamma$ is written as
$E(Z):=\inf\sigma(C_{c,Z})$.  Solovej et al \cite{Solovejetal2008}
and Frank et al \cite{Franketal2008} determined the first two terms of
the expansion of $E(Z)$ for $Z\to\infty$ and
$\gamma\leq\gamma_c$ fixed to be
\begin{equation}
  \label{eq:2.4}
  E(Z)=E^{\mathrm{TF}}(Z)+\left(\frac q4-s(\gamma)\right)Z^2+O(Z^{47/24})
\end{equation}
where
$$
s(\gamma):=\gamma^{-2}\tr\left[\left(\frac{p^2}{2}-{\gamma\over|x|}\right)_--\left(\sqrt{p^2+1}-1-\frac{\gamma}{|x|}\right)_-\right]>0
$$
is the sum of differences between the $n$-th eigenvalues of
$$\left(-\tfrac12\Delta - {\gamma\over|x|}\right)\otimes \mathbbm{1}_{\cz^q}\
\text{and}\
\left(\sqrt{-\Delta+1}-1-{\gamma\over|x|}\right)\otimes\mathbbm{1}_{\cz^q}$$
and $E^{\mathrm{TF}}(Z)$ is the infimum of the atomic Thomas--Fermi
functional $\cE_Z^\mathrm{TF}$ on its natural domain $\cI$, i.e.,
$$
\etf:=\inf(\cE^\mathrm{TF}_Z(\cI))$$
with
$$\cE^\mathrm{TF}_Z(\rho):=\int_{\rz^3}\left(\tfrac{3}{10}\gamma_{\mathrm{TF}}\rho^{5/3}(x)-{Z\over|x|}\rho(x)\right)\dx+D(\rho,\rho)$$
and
$$\cI:=\{\rho\in L^{5/3}(\rz^3)\big|\ D(\rho,\rho)<\infty,\ \rho\geq0\}.$$
Here $\gamma_{\mathrm{TF}}:=(6\pi^2/q)^{2/3}$ is the Thomas--Fermi
constant and $D(\rho,\rho)$ is the electro-static selfenergy of the
charge density $\rho$, i.e.,
$$D(\rho,\sigma)=\frac12\int_{\rz^3}\int_{\rz^3}\frac{\overline{\rho(x)}\sigma(y)}{|x-y|}\dx\dy.$$
Note that $D$ defines a scalar product and thus induces a norm, the
so-called Coulomb norm $\|\rho\|_C:=D(\rho,\rho)^{1/2}$. The minimizer
of $\cE_Z^{\mathrm{TF}}$ is denoted by $\rho_Z^{\mathrm{TF}}$.  It
obeys the scaling relation
$\rho_Z^{\mathrm{TF}}(x)=Z^2\rho_1^{\mathrm{TF}}(Z^{1/3}x)$ where
$\rho_1^{\mathrm{TF}}$ is the Thomas--Fermi density of hydrogen, i.e.,
$Z=1$ (Gombas \cite{Gombas1949}). These scaling relations and the
leading order of $E(Z)$ show that the Thomas--Fermi theory is
energetically correct in leading order even, if relativistic effects
are taken into account. Our result on the convergence of the ground
state density shows that it is also a valid model for the density on
this length scale.

We write
\begin{equation}
  \label{rhohut}
  \hat\rho_Z(x):=Z^{-2}\rho_Z(Z^{-1/3}x)
\end{equation}
for the rescaled quantum density on the Thomas--Fermi scale. This
allows to formulate our main observation:
\begin{theorem}
  \label{Thm:2.1}
  Let $\gamma\in(0,\gamma_c]$, then
  $\hat\rho_Z\to\rho_1^\mathrm{TF}$ in Coulomb norm. In fact,
  $$
  \|\hat\rho_Z-\rho_1^{\mathrm{TF}}\|_C= O(Z^{-3/16})
  $$
  as $Z\to\infty.$
\end{theorem}
Before proving this claim, we remark that the Schwarz inequality
implies also weak convergence: suppose $\sigma$ has finite Coulomb
norm, i.e., $\|\sigma\|_C<\infty$. Then
$$D(\sigma,\hat\rho_Z-\rho_1^{\mathrm{TF}})=O(Z^{-3/16}).$$
(Note that the Hardy--Littlewood--Sobolev inequality ensures that this
is the case for all $\sigma\in L^{6/5}(\rz^3)$ but that this is not
exhaustive. For example, $\sigma$ might also be a uniform charge
distribution on a sphere.)

Finally, setting $\sigma := -(1/4\pi)\Delta U$ with $U$ vanishing at
infinity gives
$$ \int U\rho \to \int U\rho_1^\mathrm{TF}\ \text{as}\ Z\to\infty$$
for all such $U$.

\begin{proof}[Proof of Theorem \ref{Thm:2.1}]
  The basic observation is, that also in this case -- as in the
  non-relativistic case done by Fefferman and Seco
  \cite{FeffermanSeco1989} -- it is useful to keep some positive term
  in the lower bound in the proof of an asymptotic energy formula:
  tracing the lower bound, the proof of the Scott conjecture by Frank
  et al does not only give the Scott formula \eqref{eq:2.4}. If one
  does not drop the positive term in Onsager's inequality -- unlike as is
  done there, we get for fixed $\gamma\in(0,2/\pi]$ the two bounds
   \begin{multline}
     \etf+\left(\frac q4-s(\gamma)\right)Z^2+\|\rho_Z^{\mathrm{TF}}-\rho_{Z}\|_C^2
     -\const Z^{47/24}\\
     \leq E(Z) \leq \etf+\left(\frac
       q4-s(\gamma)\right)Z^2+\const Z^{47/24}.
   \end{multline}
   We observe that the left and right side have identical terms up to
   order $Z^2$. Subtracting them and rearranging gives
   \begin{equation}
     \label{unskaliert}
     \|\rho_Z^{\mathrm{TF}}-\rho_{Z}\|_C^2     \leq \const Z^{47/24}.
   \end{equation}
   Since $\rho_Z^\mathrm{TF}(x)=Z^2\rho_1^\mathrm{TF}(Z^{1/3}x)$ and
   by definition of $\hat\rho_Z$ in \eqref{rhohut}, we obtain by
   change of variables
   \begin{multline}
     \|\rho_Z^\mathrm{TF}-\rho_Z\|_C^2 =\frac12 \int \dx \int dy {(\rho_Z^\mathrm{TF}(x)-\rho_Z(x))(\rho_Z^\mathrm{TF}(y)-\rho_Z(y)\over |x-y|}\\
     =\frac{Z^{7/3}}2 \int \dx \int dy
     {(\rho_1^\mathrm{TF}(x)-\hat\rho_1(x))(\rho_1^\mathrm{TF}(y)-\hat\rho_1(y)\over
       |x-y|}.
   \end{multline}
   Combining this with \eqref{unskaliert}, dividing by $Z^{7/3}$, and
   taking the root gives the claimed convergence.
 \end{proof}

 We conclude with two remarks:

 1. The proof of Solovej et al \cite{Solovejetal2008} has the same
 property as the one used here and yields a generalization for the
 multi-center case when the distance between nuclei are kept on the
 Thomas--Fermi scale.

 2. Also the proof of the Scott conjecture of the two more elaborate
 models of atoms, the Brown--Ravenhall operator treated in
 \cite{Franketal2008H} and the no-pair operator in the Furry picture
 treated in \cite{HandrekSiedentop2015}, have the same property that
 the missing error term in Onsager's inequality can be
 added. Repeating the same argument gives the analogues of Theorem
 \ref{Thm:2.1}. One merely needs to adapt the range of allowed
 constants $\gamma$ to $(0,2/(\pi/2+2/\pi)]$ and $(0,1)$ respectively
 and change the meaning of $\hat\rho_Z$ to the respective ground state
 densities.

 \textsc{Acknowledgment:} {\it Partial support of the DFG, grant SI
   348/15-1, is gratefully acknowledged.}

$\bibliographystyle{plain}
$\bibliography{coulomb}

\def\cprime{$'$}

\end{document}